\documentclass{article}

\usepackage{arxiv}

\newcommand{\sgn}{\operatornamewithlimits{sgn}}

\newcommand{\atan}{\operatornamewithlimits{atan2}}
\newcommand{\diagonal}{\operatornamewithlimits{diag}}

\newtheoremstyle{preprintstyle}
  {10pt plus 2pt minus 2pt}
  {10pt plus 2pt minus 2pt}
  {\itshape}
  {}
  {\color{secondary}\bfseries}
  { - }
  {4pt}
  {}

\theoremstyle{preprintstyle}
\newtheorem{theorem}{\textbf{Theorem}}

\newtheorem{lemma}{\textbf{Lemma}}

\newtheorem{assumption}{\textbf{Assumption}}

\title{A Robust Sensorless Controller-Observer Strategy for PMSMs with Unknown Resistance and Mechanical Model}

\author{Alessandro~Bosso
		\thanks{ A. Bosso, A. Tilli and C. Conficoni are within the research group Advanced Control and Technologies for Enhanced Mechatronics and Automation (\href{https://dei.unibo.it/en/research/research-groups/actema}{\color{secondary} ACTEMA}) at the Department of Electrical, Electronic and Information Engineering (\href{https://dei.unibo.it/en}{\color{secondary} DEI}), University of Bologna, Viale Risorgimento 2, 40136 Bologna, Italy.
		\newline Email: [\href{mailto:alessandro.bosso3@unibo.it}{alessandro.bosso3},
		\href{mailto:andrea.tilli@unibo.it}{andrea.tilli},
		\href{mailto:christian.conficoni3@unibo.it}{christian.conficoni3}]@unibo.it
		 }
		\And Andrea~Tilli \footnotemark[1]
		\And Christian~Conficoni \footnotemark[1]}
	
\copyrightTransfer{\vskip -0.3in This work has been submitted to IFAC for possible publication}
\headerAuthor{A. Bosso, A. Tilli and C. Conficoni}

\begin{document}
\maketitle

\begin{abstract}
In this work, we present a mixed sensorless strategy for Permanent Magnets Synchronous Machines, combining a torque/current controller and an observer for position, speed, flux and stator resistance.
The proposed co-design is motivated by the need of an appropriate signal injection technique, in order to guarantee full state observability.
Neither the typical constant or slowly-varying speed assumptions, nor a priori mechanical model information are used in the observer design.
Instead, the rotor speed is modeled as an unknown input disturbance with constant (unknown) sign and uniformly non-zero magnitude.
With the proposed architecture, it is shown that the torque tracking and signal injection tasks can be achieved and asymptotically decoupled.
Because of these features, we refer to this strategy as a sensorless controller-observer with no mechanical model.
Employing a gradient descent resistance/back-EMF estimation, combined with the unit circle formalism to describe the rotor position, we rigorously prove regional practical asymptotic stability of the overall structure, with a domain of attraction that can be made arbitrarily large, not including a lower dimensional manifold.
The effectiveness of this design is further validated with numerical simulations, related to a challenging application of UAV propellers control.
\end{abstract}

\keywords{Nonlinear Observers and Filter Design \and Input and Excitation Design \and Adaptive Control \and Time-Varying Systems}

\section{Introduction}

Permanent Magnet Synchronous Machines (PMSMs) are key elements for several industrial applications.
In order to enhance their reliability and minimize costs, a significant research effort has been dedicated to sensorless control techniques: these solutions exploit only electrical measurements, and knowledge of the system dynamics, to replace direct sensing of the rotor mechanical variables, i.e. angular speed and position.
Such quantities are properly estimated to implement control algorithms ensuring high regulation performance, such as those based on the field orientation principle \cite{LipoBook}.
\par Relatively to the vast literature on the topic, the present work belongs to the framework of Model-based strategies, where the machine nominal dynamics is exploited in the control design.
We refer to \cite{acarnley2006}, \cite{MarinoIJA08}, \cite{Bolognani2003} and references therein for relevant contributions in this field. 
In this respect, two particularly challenging and still not fully closed issues are: 
\begin{itemize}
\item to handle variable speed scenarios with limited/absent knowledge of the mechanical dynamics;
\item to ensure robustness to parametric uncertainties of the electromagnetic dynamics.
In fact, both PMSM stator resistance and magnetic flux significantly vary with the temperature.
\end{itemize}
Without aiming to be exhaustive, we recall some recent results containing significant steps towards the solution of the above challenges.
In \cite{Verrelli2017_CEP}, \cite{VerrelliAutomatica2018} we find sensorless control solutions with resistance reconstruction, considering time-varying speed scenarios.
There, local results are rigorously drawn, yet the mechanical model is assumed known.
In \cite{Ortega2015} interesting results are given for the case of variable speed and no use of the mechanical model, but with known stator resistance.
The authors proposed a clever reparameterization of the PMSM electromagnetic dynamics to obtain a linear regression problem, solved with standard gradient descent techniques.
The approach has been extended in \cite{bobtsovIJC2017}, where the so-called DREM filtering techniques \cite{DREM2017} are applied to improve the observer performance, and in \cite{bobtsovECC2018} where the authors include stator resistance estimation, even though they also use the mechanical model and its parameters.
The gradient-based sensorless observer described in \cite{Astolfi2010}, and modified in \cite{Praly2012} to obtain global properties, has been enhanced in  \cite{BernardIFAC2017}, with the estimation of the rotor flux amplitude.
In \cite{BernardHAL2019}, a thorough observability analysis of the pair rotor position/stator resistance is carried out.
Then, a nonlinear Luenberger observer is combined with a resistance estimator based on a scalar minimization problem.
This way, the most likely value is selected among multiple solutions in case of indistinguishability.
This strategy does not rely on the mechanical dynamics, yet known rotor flux amplitude is assumed, and speed is recovered by filtering the position estimate.
In \cite{SCL2019} we proposed a speed, position, and flux sensorless observer solution, with no information about the mechanical model.
By Lyapunov-based and two time scales arguments, regional practical asymptotic stability was established for the case of time-varying speed with bounded derivative and unknown constant sign.
In \cite{bosso2020_semiglobal}, a hybrid modification of the aforementioned observer is introduced to achieve semi-global stability and a faster estimation dynamics.
\par The present work expands significantly the aforementioned approach. Indeed, the hypothesis of perfectly known electrical parameters is removed by letting the stator resistance be uncertain.
Concerning the observability properties of this additional parameter, relatively strong signal conditions must be satisfied \cite{BernardHAL2019}, and these are related to the behavior of stator currents.
In this respect, a complete separation between controller and observer design cannot guarantee, a priori, desirable tracking and stability performance.
For this purpose, we take a step towards a more coupled sensorless strategy, where an observer, inspired by the solution proposed in \cite{SCL2019}, is interconnected with a torque/current controller.
This controller is responsible not only for standard torque generation objectives, but also for signal injection to ensure resistance observability.
In practice, the proposed solution should be then analyzed in closed loop with a speed controller.
However, we decide to simplify the mathematical framework by supposing that such feedback does not destroy the conditions on speed for sensorless observability, so that we can focus on the controller-observer structure.
\par To summarize the goals of the present work, we want to achieve reconstruction of rotor speed, position, flux and stator resistance, while at the same time ensuring that the torque requests given by an upper level (a priori given) speed controller are satisfied asymptotically.
The two time scales separation of \cite{SCL2019} is used to simplify the stability analysis, and to the same slow subsystem (an attitude observer on the unit circle), we associate a fast subsystem which now includes a current controller with adaptive resistance and back-EMF reconstruction.
An Immersion and Invariance (I\&I) strategy \cite{astolfi2003immersion} is proposed to estimate the stator resistance and yield a simple continuous-time gradient descent algorithm.
In this context, we show that by tracking a sinusoidal current reference, not responsible for electric torque distortion, it is possible to derive regional practical asymptotic stability, with the domain of attraction that can be extended arbitrarily in size, except for a lower-dimensional unstable manifold.
\par The paper is organized as follows. 
In Section \ref{sec:ProblStat} we briefly recall the PMSM model, so that we can state the observer-controller problem and introduce the required standing assumptions.
Section \ref{sec:ObsScheme} is devoted to the proposed solution and its stability analysis.
The results are further validated in Section \ref{sec:simulation} by means of numerical simulation tests.
Finally, Section \ref{sec:conclusion} wraps up the paper with final remarks and future directions.

\section{Notation and Useful Results}

\subsection{Notation}
We use $(\cdot)^T$ to denote the transpose of real-valued matrices, and often write $(v, w)$, with $v$, $w$ column vectors, to indicate the concatenated vector $(v^T \; w^T)^T$.
When clear from the context, the time argument will be omitted for notational simplicity.
In case of non-differentiable signals, the upper right Dini derivative, indicated with $D^+$, will be employed as generalized derivative.
\par Following the formalism in \cite{SCL2019}, we employ the unit circle $\mathbb{S}^1$ to represent reference frames involved in the manipulation of PMSM equations.
Indeed, $\mathbb{S}^1$ is a compact abelian Lie group of dimension $1$, with identity element $(1 \; 0)^T$.
A generic integrator on $\mathbb{S}^1$ is given by:
\begin{equation}\nonumber
\dot{\zeta} = u(t)\underbrace{\begin{pmatrix} 0 & -1\\1 & 0\end{pmatrix}}_{\mathcal{J}}\zeta, \qquad \zeta \in \mathbb{S}^1,
\end{equation}
with $u(t) \in \mathbb{R}$.
Any angle $\vartheta \in \mathbb{R}$ can be mapped into an element of the unit circle given by $(\cos(\vartheta) \; \sin(\vartheta))^T \in \mathbb{S}^1$.
Furthermore, to any $\zeta = (c \; s)^T \in \mathbb{S}^1$ we can associate a rotation matrix $\mathcal{C}[\zeta] = {\tiny \begin{pmatrix} c & -s\\ s & c \end{pmatrix}}$, which is used for group multiplication: for any $\zeta_1$, $\zeta_2 \in \mathbb{S}^1$, their product is given by $\mathcal{C}[\zeta_1]\zeta_2 = \mathcal{C}[\zeta_2]\zeta_1$.
\subsection{Uniform Complete Observability}
Consider a linear time-varying system of the form:
\begin{equation}\nonumber
\dot{x} = A(t)x \qquad y = C(t)x.
\end{equation}
with $x \in \mathbb{R}^n$, $y \in \mathbb{R}^m$ and $A: \mathbb{R}_{\geq 0} \to \mathbb{R}^{n\times n}$, $C: \mathbb{R}_{\geq 0} \to \mathbb{R}^{m \times n}$ piecewise continuous functions.
Denote with $\Phi$ the transition matrix associated with $A$, then we say that the pair $(C, A)$ is uniformly completely observable (UCO) if there exist positive scalars $\alpha_1$, $\alpha_2$, $\delta$ such that, for all $t \geq 0$ \cite{sastry2011adaptive}:
\begin{equation}\nonumber
\alpha_1 I_n \leq \int_{t}^{t + \delta}\Phi^T(s, t)C^T(s)C(s)\Phi(s, t)ds \leq \alpha_2I_n,
\end{equation}
where the above integral is known as observability Gramian.
In classic continuous-time system identification, the UCO property is the key element to achieve exponential stability of the parameter estimation error.
Indeed, the system
\begin{equation}\nonumber
\dot{x} = -\psi(t)\psi^T(t)x,
\end{equation}
with $x \in \mathbb{R}^n$ and $\psi: \mathbb{R}_{\geq 0} \to \mathbb{R}^n$ a piecewise continuous function, can be proven to be uniformly globally exponentially stable as long as $\psi$ is persistently exciting (PE), i.e. the pair $(\psi^T, 0_{n\times n})$ is UCO.
The stability proof of the observer will involve similar arguments in a slightly modified context.

\section{Model Formulation and Problem Statement}\label{sec:ProblStat}

Firstly, we recall the PMSM model.
Under balanced working conditions, linear magnetic circuits, and negligible iron losses, according to standard planar representation of three-phase electrical machines, we can write the two-phase PMSM dynamics in a static frame as:
\begin{equation}\label{eq:PMSM_static}
\frac{d}{dt}i_s  = -\frac{R}{L}i_s + \frac{u_s}{L} - \frac{\omega\varphi \mathcal{J}\zeta}{L}, \qquad \dot{\zeta} = \omega \mathcal{J} \zeta
\end{equation}
where $i_s$, $u_s \in \mathbb{R}^2$ are the stator currents and voltages, $\omega$ is the electrical rotor speed, while $\zeta \in \mathbb{S}^1$ and $\varphi \in \mathbb{R}_{>0}$ represent the orientation and amplitude of the rotor flux vector.
In addition, $R$ and $L$ are the stator resistance and inductance, respectively.
Finally, the torque generated by the PMSM is given by:
\begin{equation}
T_{\text{el}} = -\frac{3}{2}p\varphi \zeta^T\mathcal{J}i_s,
\end{equation}
with $p \in \mathbb{Z}_{\geq 1}$ the number of pole pairs of the motor.
Consider a generic rotating reference frame, $\zeta_r \in \mathbb{S}^1$, satisfying:
\begin{equation}
\dot{\zeta}_r = \omega_r \mathcal{J}\zeta_r,
\end{equation}
with $\omega_r$ a piecewise continuous signal.
Then, system \eqref{eq:PMSM_static} can be rewritten in the rotating frame $\zeta_r$ as:
\begin{equation}\label{eq:PMSM_dyn_gen}
\frac{d}{dt}i_r = -\frac{R}{L}i_r+\frac{u_r}{L} - \frac{\omega\varphi \mathcal{J}\mathcal{C}^T[\zeta_r]\zeta}{L} - \omega_r\mathcal{J}i_r, \qquad \dot{\zeta} = \omega \mathcal{J} \zeta
\end{equation}
where $i_r = \mathcal{C}^T[\zeta_r]i_s$, $u_r = \mathcal{C}^T[\zeta_r]u_s$.\\
Concerning the observability properties of system \eqref{eq:PMSM_dyn_gen}, the assumption of non-permanent zero speed is usually required as sufficient condition to reconstruct $\omega$, $\zeta$, $\varphi$, assuming currents and voltages available for measurement and the parameters $R$ and $L$ perfectly known \cite{zaltni2010synchronous}.
However, this is in general not sufficient to reconstruct the stator resistance, so additional features, in particular related to the stator currents (see \cite{BernardHAL2019} and references therein), need to be satisfied.
The fundamental idea is that a sufficiently rich current profile can be used to impose persistency of excitation, thus guaranteeing observability of $R$.
\par The desired richness features are not a priori guaranteed, though, if control and observation tasks are completely decoupled.
As a consequence, the focus of this work is on the co-design of both a sensorless observer and a current controller that, altogether, formally ensure robust observability, under some commonly satisfied working conditions.
We aptly refer to this co-design as a \textit{mixed controller-observer problem}.
It must be stressed that the speed controller is a priori given, as it is not directly responsible for current excitation, and its main role in the present work is to provide a reference for torque regulation, indicated with $T^*_{\text{el}}$.
Still, a feedback between the controller-observer and the speed dynamics clearly exists, and an incorrect behavior of the rotor speed can potentially destroy the observability properties.
To simplify the mathematical setup, we assume that the observability properties deriving from $\omega$ are a priori satisfied.
This, in practice, translates to imposing specific tracking properties to the controller-observer, and a strong connection with the domain of attraction arises.
We leave the analysis of this interconnection out of the scope of this work.
\par For observer design purposes, $\omega$ is modeled as an unknown input which is supposed to satisfy the following regularity Assumption.
\begin{assumption}\label{hyp:omega}
The signal $\omega(\cdot)$ is defined on the interval $[0, \infty)$\footnote{The initial time is chosen to be $0$, without loss of generality, since we assume invariance of the properties of $\omega$.} and, in addition:
\begin{enumerate}[a)]
\item $\omega(\cdot)$ is $\mathcal{C}^0$ and piecewise $\mathcal{C}^1$ in its domain of existence;
\item there exist positive scalars $\omega_{\min}$, $\omega_{\max}$ such that, for all $t \geq 0$, it holds $\omega_{\min} \leq |\omega(t)| \leq \omega_{\max}$;
\item $|D^+\omega(t)|$ exists and is bounded, for all $t \geq 0$. 
\end{enumerate}
\end{assumption}
The constant sign and uniformly non-zero magnitude assumptions are indeed somewhat restrictive, and are the same employed in \cite{SCL2019}.
However, significant applications such as renewables electric energy generation and electric vehicles propulsion (UAVs, HEVs) usually meet such requirements.
Finally, we present an Assumption used for design purposes.
\begin{assumption}\label{hyp:info}
The following hypotheses hold:
\begin{enumerate}[a)]
\item the torque reference $T^*_{\textup{el}}(\cdot)$ is $\mathcal{C}^1$ in the interval $[0, \infty)$ and satisfies $\| T^*_{\textup{el}}(\cdot) \|_{\infty} \leq T^*_{\max}$, for some positive scalar $T^*_{\max}$.
Furthermore $T^*_{\textup{el}}(\cdot)$ and its derivative, $\dot{T}^*_{\textup{el}}(\cdot)$, are available for measurement;
\item $u_s$ and $i_s$ are available for control, along with parameters $L$ and $p$;
\item $\zeta$, $\omega$, $\varphi$ and $R$ are unknown.
\end{enumerate}
\end{assumption}
We can thus summarize the \textit{mixed controller-observer problem} as follows.
Given the PMSM dynamics \eqref{eq:PMSM_static} or \eqref{eq:PMSM_dyn_gen} and under the hypothesis that Assumptions \ref{hyp:omega}-\ref{hyp:info} are satisfied, design a dynamical system such that:
\begin{itemize}
\item an estimate of $\zeta$, $\omega$, $\varphi$ and $R$ is provided, ensuring appropriate stability and convergence properties;
\item the PMSM torque, $T_{\text{el}}(\cdot)$, reaches the reference $T^*_{\text{el}}(\cdot)$ (practically) asymptotically.
\end{itemize}
While the first requirement resembles a typical problem of sensorless observation, the second means that the current controller should not deteriorate, at least asymptotically, the torque tracking performance.
Note that this torque-preservation property is possible because, for any $t \geq 0$, only one stator current direction in $\mathbb{R}^2$ generates torque, so that the exciting signal can be in principle completely masked.

\section{The Proposed Scheme}\label{sec:ObsScheme}

Let $\chi \coloneqq |\omega|\varphi \in \mathbb{R}_{>0}$, $\xi \coloneqq (1/\varphi)\sgn(\omega)$, then let $\zeta_\chi \coloneqq \zeta \sgn(\xi) = \zeta \sgn(\omega)$.
Replacing $\zeta$ with $\zeta_\chi$ in \eqref{eq:PMSM_dyn_gen}, we obtain the following system:
\begin{equation}\label{eq:PMSM_chi_frame}
\begin{split}
\frac{d}{dt}i_r & = -\frac{R}{L}i_r+\frac{u_r}{L} - \frac{\chi \mathcal{J}\mathcal{C}^T[\zeta_r]\zeta_\chi}{L} - \omega_r\mathcal{J}i_r \\
\dot{\zeta}_\chi & = \chi\xi \mathcal{J} \zeta_\chi.
\end{split}
\end{equation}
In this structure, $\xi \in \mathbb{R}$ is an unknown parameter and $\chi(\cdot)$ a positive signal which, by virtue of Assumption \ref{hyp:omega}, is such that:
\begin{enumerate}[a)]
\item $\chi(\cdot)$ is $\mathcal{C}^0$ and piecewise $\mathcal{C}^1$;
\item $\chi_{\text{m}} \leq \chi(t) \leq \chi_{\text{M}}$, for all $t \in [0, \infty)$ and for some positive scalars $\chi_{\text{m}}$, $\chi_{\text{M}}$;
\item $|D^+\chi(t)| \leq M$, for all $t \in [0, \infty)$ and for some positive scalar $M$.
\end{enumerate}
Consider now a reference frame, $\zeta_r = \hat{\zeta}_{\chi}$, used to estimate the frame $\zeta_\chi$, satisfying:
\begin{equation}
\dot{\hat{\zeta}}_\chi = \hat{\omega}_\chi \mathcal{J}\hat{\zeta}_\chi,
\end{equation}
with $\hat{\omega}_\chi$ to be designed for control purposes.
Let $\eta \coloneqq \mathcal{C}^T[\hat{\zeta}_\chi]\zeta_\chi \in \mathbb{S}^1$ be the synchronization error, and let $i_{\hat{\chi}}$, $u_{\hat{\chi}}$ be the electric variables in the frame $\hat{\zeta}_\chi$.
Then, we can rewrite system \eqref{eq:PMSM_chi_frame}, along with the alignment error dynamics, as:
\begin{equation}\label{eq:PMSM_chi_frame_eta}
\begin{split}
\frac{d}{dt}i_{\hat{\chi}} & = -\frac{R}{L}i_{\hat{\chi}}+\frac{u_{\hat{\chi}}}{L} - \frac{\chi \mathcal{J}\eta}{L} - \hat{\omega}_\chi\mathcal{J}i_{\hat{\chi}} \\
\dot{\zeta}_\chi & = \chi \xi \mathcal{J} \zeta_\chi, \quad \dot{\eta} = (\chi\xi - \hat{\omega}_\chi) \mathcal{J} \eta = \omega_\eta\mathcal{J}\eta.
\end{split}
\end{equation}
The electric torque is related to the currents, in the rotating frame $\hat{\zeta}_{\chi}$, as follows:
\begin{equation}
\zeta_{\chi}^T \mathcal{C} [\hat{\zeta}_{\chi}]\mathcal{J} i_{\hat{\chi}} = \eta^T \mathcal{J} i_{\hat{\chi}} = -\eta_1i_{\hat{\chi}2} + \eta_2i_{\hat{\chi}1} = -\frac{2}{3p} \xi T_{\text{el}},
\end{equation}
therefore, if the frames $\hat{\zeta}_\chi$ and $\zeta_\chi$ achieve synchronization, it holds $i_{\hat{\chi}2} = 2/(3p)\xi T_{\text{el}}$, which corresponds to the typical expression employed in sensored field-oriented control.
The remaining component, $i_{\hat{\chi}1}$, can be then freely assigned in order to achieve the desired signal injection.
This suggests that if the synchronization problem is solved, then the solution of the torque tracking objective is in turn satisfied, as long as  $i_{\hat{\chi}1}$ and $i_{\hat{\chi}2}$ are used for signal injection and torque generation, respectively.
\subsection{An Adaptive Attitude Observer on $\mathbb{S}^1$}
Indicate with $h \coloneqq -\chi\mathcal{J}\eta$ the back-EMF vector, and notice that $|h| = \chi$.
Let $\hat{h} = (\hat{h}_1, \hat{h}_2) \in \mathbb{R}^2$, $\hat{\xi} \in \mathbb{R}$ be appropriate estimates of $h$ and $\xi$, respectively, and consider the adaptive observer
\begin{equation}\label{eq:attitude_obs}
\dot{\hat{\zeta}}_\chi = (|\hat{h}|\hat{\xi} + k_\eta\hat{h}_1)\mathcal{J}\hat{\zeta}_\chi, \qquad \dot{\hat{\xi}} = \gamma\hat{h}_1,
\end{equation}
for some positive scalars $k_\eta$, $\gamma$.
In a certainty-equivalence sense, suppose that $h$ is perfectly known, then we can rewrite the dynamics of $x_{\text{s}} \coloneqq (\eta, \tilde{\xi}) \in \mathbb{S}^1 \times \mathbb{R}$, with $\tilde{\xi} \coloneqq \xi - \hat{\xi}$ as follows:
\begin{equation}\label{eq:red_ord_syst}
\dot{\eta} = \left(\chi \tilde{\xi} - k_{\eta}\chi\eta_2 \right)\mathcal{J}\eta, \qquad \dot{\tilde{\xi}} = -\gamma\chi\eta_2.
\end{equation}
In \cite{SCL2019} it was proven that under Assumption \ref{hyp:omega} the attractor $\bar{x}_{\text{s}}  = ((1, 0), 0)$ of system \eqref{eq:red_ord_syst} is uniformly asymptotically stable, with a domain of attraction that can be extended to all points in $\mathbb{S}^1 \times \mathbb{R}$, except for a lower dimensional unstable manifold, originating from the saddle equilibrium $((-1, 0), 0)$.
Furthermore, when $\hat{h} \neq h$ it can be shown that it holds:
\begin{equation}\label{eq:red_ord_syst_pert}
\dot{x}_{\text{s}}  = f_{\text{s}} (x_{\text{s}} , \chi) + N(x_{\text{s}} , \tilde{h}, \chi, \xi),
\end{equation}
with $\tilde{h} \coloneqq h - \hat{h}$ the back-EMF estimation error, $N(\cdot)$ a continuous function vanishing in $\tilde{h} =  0$, and $f_{\text{s}} (\cdot)$ coinciding with the vector field in \eqref{eq:red_ord_syst}.
These considerations suggest that $\hat{h}(t)$ should be designed so that the trajectories of system \eqref{eq:red_ord_syst_pert} are ``close'' to those of system \eqref{eq:red_ord_syst}.
With a time scale separation approach, we thus decide to design the current controller-observer, with resistance and back-EMF estimation, as the fast subsystem of the structure.
\subsection{Fast Subsystem Design Steps}
Consider a torque reference $T^*_{\text{el}}(\cdot)$ satisfying Assumption \ref{hyp:info}, then define the following signals:
\begin{equation}
i_{\text{q}}^* = \frac{2}{3p}\hat{\xi} T^*_{\text{el}}, \qquad p_{\text{q}}^* = \frac{2}{3p} \left(\dot{\hat{\xi}} T^*_{\text{el}} + \hat{\xi} \dot{T}^*_{\text{el}} \right).
\end{equation}
We can summarize the proposed strategy as follows:
\begin{itemize}
\item an observer of the stator current $i_{\hat{\chi}}$ is designed, including a suitable adaptive law for $R$ and $h$ (both regarded in this step as constant parameters);
\item the current estimate is imposed to track a reference of the form $i^*_{\hat{\chi}} = (w_1, i_\text{q}^*)$, where $w_1$ is the output of an exosystem used for sinusoidal generation;
\item the time scale separation is imposed by suitably choosing the gains of the structure, thus restoring the desirable behavior of the adaptive attitude observer \eqref{eq:attitude_obs}, as highlighted in the comparison between \eqref{eq:red_ord_syst} and \eqref{eq:red_ord_syst_pert}.
\end{itemize}
This way, the observer problem is solved and, as a consequence of frames synchronization, torque tracking is achieved.
\subsection{Indirect I\&I Adaptive Current Controller-Observer}
For convenience, rewrite the current dynamics \eqref{eq:PMSM_chi_frame_eta} as a linear regression form:
\begin{equation}
\begin{split}
&\frac{d}{dt}i_{\hat{\chi}} = L^{-1}\Omega^T(i_{\hat{\chi}})\theta + L^{-1}u_{\hat{\chi}} - \hat{\omega}_{\chi}\mathcal{J}i_{\hat{\chi}}, \quad \theta = \begin{pmatrix} R \\ h \end{pmatrix}\\
&D^+{\theta} = f_\theta(\chi, D^+\chi, \eta, \omega_\eta), \quad \Omega^T(i_{\hat{\chi}}) = \begin{pmatrix} - i_{\hat{\chi}} & I_2 \end{pmatrix}.
\end{split}
\end{equation}
If the map $f_\theta$ is identically zero, the above dynamics can be treated as in classic adaptive observer design, with $\theta$ an unknown parameter vector.
Bearing this idea in mind, consider an I\&I observer of the form
\begin{equation}
\begin{split}
&\dot{\hat{\imath}} = L^{-1}\Omega^T(i_{\hat{\chi}})(\hat{\theta} + \beta(i_{\hat{\chi}})) + L^{-1}u_{\hat{\chi}} - \hat{\omega}_\chi\mathcal{J}i_{\hat{\chi}} + k_{\text{p}}(i_{\hat{\chi}} - \hat{\imath})\\
&\dot{\hat{\theta}} = -L^{-1}\frac{\partial \beta}{\partial i_{\hat{\chi}}}(i_{\hat{\chi}}) \left[\Omega^T(i_{\hat{\chi}})(\hat{\theta} + \beta(i_{\hat{\chi}})) + u_{\hat{\chi}} - L\hat{\omega}_\chi\mathcal{J}i_{\hat{\chi}} \right],
\end{split}
\end{equation}
where $k_{\text{p}}$ is a positive scalar, while $\beta(i_{\hat{\chi}})$ is a map to be defined in the following.
Let $\tilde{\imath} = i_{\hat{\chi}} - \hat{\imath}$, $z = \hat{\theta} + \beta(i_{\hat{\chi}}) - \theta$, so that the resulting error dynamics becomes
\begin{equation}
\begin{split}
\dot{\tilde{\imath}} &= -L^{-1}\Omega^T(i_{\hat{\chi}})z - k_{\text{p}}\tilde{\imath}\\
D^+{z} &= -L^{-1}\frac{\partial \beta}{\partial i_{\hat{\chi}}}(i_{\hat{\chi}})\Omega^T(i_{\hat{\chi}})z - f_\theta(\chi, D^+\chi, \eta, \omega_\eta),
\end{split}
\end{equation}
thus suggesting the choice $\partial\beta/\partial i_{\hat{\chi}} = k_z\Omega$, with $k_z$ a positive scalar, therefore
\begin{equation}
\beta(i_{\hat{\chi}}) = k_z\begin{pmatrix}
-\frac{|i_{\hat{\chi}}|^2}{2}\\i_{\hat{\chi}}
\end{pmatrix}.
\end{equation}
It follows that the parameter estimation error takes the form
\begin{equation}
D^+{z} = -(k_z L^{-1})\underbrace{\left[\Omega(i_{\hat{\chi}})\Omega^T(i_{\hat{\chi}})\right]}_{M(i_{\hat{\chi}})}z - f_\theta(\chi, D^+\chi, \eta, \omega_\eta),
\end{equation}
which corresponds, for $f_\theta = 0$, to a classical gradient descent algorithm.
The parameter estimates are then given by:
\begin{equation}
\begin{split}
\hat{R} &= \begin{pmatrix}1 & 0_{1\times2} \end{pmatrix}\hat{\theta} - (k_z/2)|i_{\hat{\chi}}|^2\\
\hat{h} &= \begin{pmatrix}0_{2\times1} & I_2 \end{pmatrix}\hat{\theta} + k_z i_{\hat{\chi}}.
\end{split}
\end{equation}
\\
Consider the exosystem
\begin{equation}
\frac{d}{dt}\begin{pmatrix} w_1\\ w_2 \end{pmatrix} = \begin{pmatrix} 0 & \lambda\\ -\lambda & 0 \end{pmatrix}\begin{pmatrix} w_1\\ w_2 \end{pmatrix},
\end{equation}
with $\lambda$ a positive scalar for tuning.
Instead of working with a tracking error defined with the measured currents, we consider the estimated current mismatch $e =\hat{\imath} - i^*_{\hat{\chi}}$, suggesting a proportional controller of the form:
\begin{equation}
u_{\hat{\chi}} = -\Omega^T(i_{\hat{\chi}})(\hat{\theta} + \beta(i_{\hat{\chi}})) + L\hat{\omega}_{\chi}\mathcal{J}i_{\hat{\chi}} -Lk_e e + L\begin{pmatrix} \lambda w_2 \\ p_{\text{q}}^*\end{pmatrix},
\end{equation}
which in turn leads to the tracking error dynamics:
\begin{equation}
\dot{e} = -k_e e + k_{\text{p}}\tilde{\imath}.
\end{equation}
As a consequence, write the overall error system as
\begin{equation}\label{eq:ctrl_obs_error}
\begin{split}
\frac{d}{dt}\begin{pmatrix} w_1\\ w_2 \end{pmatrix} &= \begin{pmatrix} 0 & \lambda\\ -\lambda & 0 \end{pmatrix}\begin{pmatrix} w_1\\ w_2 \end{pmatrix} \qquad i^*_{\hat{\chi}}(w, t) = \begin{pmatrix}w_1 \\ i_\text{q}^*(t) \end{pmatrix}\\
\frac{d}{dt}\begin{pmatrix}e \\ \tilde{\imath} \end{pmatrix} &= \begin{pmatrix} -k_e I_2 & k_{\text{p}} I_2\\ 0_{2\times2} & -k_{\text{p}} I_2 \end{pmatrix}\begin{pmatrix}e \\ \tilde{\imath} \end{pmatrix} - \begin{pmatrix}0_{2\times2} \\ L^{-1}I_2\end{pmatrix}\Omega^T(i_{\hat{\chi}})z\\
D^+{z} &= -(k_z L^{-1})M(i^*_{\hat{\chi}} + e + \tilde{\imath})z - f_\theta(\chi, D^+\chi, \eta, \omega_\eta)\\
\dot{x}_{\text{s}} &= f_{\text{s}}(x_{\text{s}}, \chi) + N(x_{\text{s}}, \tilde{h}, \chi, \xi).
\end{split}
\end{equation}
Choose, and \textit{fix}, some positive scalars $\bar{\kappa}_e$, $\bar{\kappa}_{\text{p}}$, $\bar{\kappa}_z$.
Appealing to classical singular perturbations arguments, denote with $\varepsilon$ the perturbation parameter, let $\lambda = k_e/\bar{\kappa}_e = k_{\text{p}}/\bar{\kappa}_{\text{p}} = k_z / (L\bar{\kappa}_z) = \varepsilon^{-1}$ and let $t = t^* + \varepsilon \tau$, with $\tau$ indicating the fast time scale and $t^* \geq 0$.
The boundary layer system can be then written as follows (denote with $(\cdot)'$ the time derivative in the fast scale):
\begin{equation}\label{eq:ctrl_obs_bl}
\begin{split}
\begin{pmatrix} w_1'\\ w_2' \end{pmatrix} &= -\mathcal{J}\begin{pmatrix} w_1\\ w_2 \end{pmatrix}, \qquad i^*_{\hat{\chi}}(w, t^*) = \begin{pmatrix}w_1 \\ i_{\text{q}}^*(t^*) \end{pmatrix}\\
\begin{pmatrix}e' \\ \tilde{\imath}' \end{pmatrix} &= \begin{pmatrix} - I_2\bar{\kappa}_e & I_2 \bar{\kappa}_{\text{p}}\\ 0_{2\times2} & -I_2\bar{\kappa}_{\text{p}} \end{pmatrix}\begin{pmatrix}e \\ \tilde{\imath} \end{pmatrix} = A_{\text{f}} \begin{pmatrix}e \\ \tilde{\imath} \end{pmatrix}\\
z' &= -\bar{\kappa}_z M(i^*_{\hat{\chi}}(w, t^*) + e + \tilde{\imath})z,
\end{split}
\end{equation}
and note that only ordinary derivatives can be used since $f_\theta$ is absent.
We have the following stability result.
\begin{lemma}\label{lemma:bl_stability}
Suppose that $|i_{\textup{q}}^*(t^*)| \leq I^*$, for a positive scalar $I^*$.
Then, there exists a positive scalar $W^*$ such that, for any $w(0)$ satisfying $|w(0)| \geq W^*$, the origin of system
\begin{equation}\label{eq:bl_lemma_system}
\frac{d}{d\tau}\begin{pmatrix} e\\ \tilde{\imath} \\ z \end{pmatrix} = \begin{pmatrix}
-I_2\bar{\kappa}_e & I_2\bar{\kappa}_{\textup{p}} & 0_{2\times3}\\
0_{2\times2} & -I_2\bar{\kappa}_{\textup{p}} & 0_{2\times3}\\
0_{3\times2} & 0_{3\times2} & -\bar{\kappa}_z M(i^*_{\hat{\chi}} + e + \tilde{\imath})
\end{pmatrix}
\begin{pmatrix} e\\ \tilde{\imath} \\ z \end{pmatrix}
\end{equation}
is uniformly globally asymptotically stable and locally exponentially stable.
\end{lemma}
\begin{proof}
Due to the structure of system \eqref{eq:bl_lemma_system}, we have that the $(e, \tilde{\imath})$-subsystem is globally exponentially stable, hence there exist positive scalars $a_{1{\text{f}}}$, $a_{2{\text{f}}}$ such that:
\begin{equation}
\left| \begin{matrix} e(\tau) \\ \tilde{\imath}(\tau) \end{matrix} \right| \leq a_{1{\text{f}}}\exp \left(-a_{2{\text{f}}} \tau \right)\left| \begin{matrix} e(0) \\ \tilde{\imath}(0) \end{matrix} \right|
\end{equation}
The proof collapses then to the stability analysis of the $z$-subsystem, which is written as:
\begin{equation}
\begin{split}
&\frac{d}{dt}z = -\bar{\kappa}_z(\underbrace{\Omega_w\Omega_w^T}_{M_w(w, e, \tilde{\imath})} + \underbrace{\Omega_t\Omega_t^T}_{M_t(e, \tilde{\imath}, t^*)})z \\
\Omega_w = &\begin{pmatrix}-w_1 - e_1 - \tilde{\imath}_1\\ 1 \\ 0 \end{pmatrix}, \Omega_t = \begin{pmatrix}-i_{\text{q}}^* - e_2 - \tilde{\imath}_2\\ 0 \\ 1 \end{pmatrix} 
\end{split}
\end{equation}
with $\Omega = \begin{pmatrix} \Omega_w(w, e, \tilde{\imath}) & \Omega_t(e, \tilde{\imath}, t^*)) \end{pmatrix}$.
The main idea is to exploit the richness properties of $w$, regardless of the fact that both $\Omega_w$ and $\Omega_t$ are not individually PE.
\par Firstly, we want to prove that, for any bounded $e(0)$, $\tilde{\imath}(0)$, the pair $(\Omega^T, 0_{3\times3})$ is UCO, i.e.:
\begin{equation}
\alpha_1 \leq \int_{\tau}^{\tau + \delta} \underbrace{x^TM(i^*_{\hat{\chi}}(w(s), t^*) + e(s) + \tilde{\imath}(s))x}_{x^T \Omega_w\Omega_w^T x + x^T \Omega_t\Omega_t^T x} ds \leq \alpha_2,
\end{equation}
for some positive scalars $\alpha_1$, $\alpha_2$, $\delta$, for all $x = \begin{pmatrix}x_1 & x_2 & x_3 \end{pmatrix}^T \in \mathbb{R}^3$ satisfying $|x| = 1$.
Let $\Omega_0(\tau) = \Omega(i^*_{\hat{\chi}}(w(\tau), t^*))$ and $M_0(\tau) = \Omega_0(\tau)\Omega_0^T(\tau)$, then it follows that (for $\delta = 2\pi$)
\begin{equation}\nonumber
\begin{split}
\int_{\tau}^{\tau + \delta}x^T M_0(s) xds = &\; \int_\tau^{\tau + \delta}[x_1^2w_1^2(s) -2x_1x_2w_1(s)]ds\\
&\;+ \delta(x_1^2(i_{\text{q}}^*)^2 - 2x_1x_3 i_{\text{q}}^* + x_2^2 +x_3^2)\\
\geq &\; \pi |w(0)|^2 x_1^2 + 2\pi(x_2^2 + x_3^2)\\
&\;+ 2\pi x_1^2(i_{\text{q}}^*)^2 - 4\pi|x_1||x_3|I^*\\
\geq &\; \pi |w(0)|^2 x_1^2 + 2\pi(x_2^2 + x_3^2)\\
&\;+ 2\pi x_1^2(i_{\text{q}}^*)^2 - \frac{2\pi}{\rho} (x_1I^*)^2 - \rho 2\pi x_3^2,
\end{split}
\end{equation}
for any positive scalar $\rho$.
Choose $\rho \in (0, 1)$ and $W^*$ sufficiently large to enforce that all coefficients multiplying $x_i$, $i \in \{ 1, 2, 3 \}$, are positive.
Hence the pair $(\Omega_0^T, 0_{3\times3})$ is UCO.
Denote with $\beta_1$ the UCO lower bound of $(\Omega_0^T, 0_{3\times3})$, and let $\Omega = \Omega_0 + \Delta\Omega (e, \tilde{\imath})$, with $\Delta\Omega = \begin{pmatrix} \Delta\Omega_w& \Delta\Omega_t \end{pmatrix}$.
Following \cite[Lemma 6.1.2]{sastry2011adaptive}, apply the triangle inequality to yield:
\begin{equation}\nonumber
\begin{split}
\sqrt{\int_{\tau}^{\tau + \delta}x^T M x ds} =&\; \sqrt{\bigintsss_\tau^{\tau + \delta}\left|\begin{pmatrix}\Omega_w^T\\\Omega_t^T\end{pmatrix}x\right|^2ds}\\
=&\sqrt{\bigintsss_\tau^{\tau + \delta}\left|\Omega_0^T x + \begin{pmatrix}\Delta\Omega_w^T\\ \Delta\Omega_t^T\end{pmatrix}x \right|^2ds}\\
\geq&\sqrt{\beta_1} - \sqrt{\int_\tau^{\tau + \delta}|e + \tilde{\imath}|^2ds}\\
\geq&\sqrt{\beta_1} - 2\sqrt{\delta}\max_{s \geq \tau}\left| \begin{matrix} e(s) \\ \tilde{\imath}(s) \end{matrix} \right|\\
\geq&\sqrt{\beta_1} - 2\sqrt{\delta}a_{1{\text{f}}}\exp \left(-a_{2{\text{f}}} \tau \right)\left| \begin{matrix} e(0) \\ \tilde{\imath}(0) \end{matrix} \right|.
\end{split}
\end{equation}
Therefore, it is sufficient that
\begin{equation}\label{eq:local_pe_bound}
\left| \begin{matrix} e(0) \\ \tilde{\imath}(0) \end{matrix} \right| < \frac{1}{2a_{1{\text{f}}}} \sqrt{ \frac{\beta_1}{2\pi} }
\end{equation}
in order to yield $(\Omega^T, 0_{3\times3})$ UCO, with lower bound denoted with $\alpha_1$.
If the initial conditions do not satisfy \eqref{eq:local_pe_bound}, clearly it is sufficient to wait a finite time $T$ such that
\begin{equation}
T > \frac{1}{a_{2{\text{f}}}}\log\left( \left| \begin{matrix} e(0) \\ \tilde{\imath}(0) \end{matrix} \right| 2a_{1{\text{f}}} \sqrt{\frac{2\pi}{\beta_1}} \right)
\end{equation}
to recover the previous UCO property, with the same lower bound, a possibly higher upper bound, and $\delta = T + 2\pi$.
\par Pick an arbitrary positive scalar $c > 0$, and choose any initial condition satisfying $|(e(0), \tilde{\imath}(0))| \leq c$, with arbitrary $|z(0)|$.
We have that $(\Omega^T, 0_{3\times3})$ is UCO, with bounds depending on $c$.
By classic identification results, $(\Omega^T, 0_{3\times3})$ UCO if and only if so is $(\Omega^T, -\Omega \Omega^T)$.
Consider $V = z^Tz/(2\bar{\kappa}_z)$, then:
\begin{equation}
\int_{\tau}^{\tau + \delta}\frac{d}{ds}V(s)ds = \int_{\tau}^{\tau + \delta}-z(s)^T\Omega \Omega^Tz(s)ds \leq -\alpha|z(\tau)|^2,
\end{equation}
for some $\alpha$ such that $0 < 1 - 2\bar{\kappa}_z\alpha < 1$.
From \cite[Theorem 1.5.2]{sastry2011adaptive} we have:
\begin{equation}
|z(\tau)| \leq \sqrt{\frac{1}{1-2\bar{\kappa}_z\alpha}}\exp\left(-\frac{1}{2\delta} \log\left(\frac{1}{1-2\bar{\kappa}_z\alpha}\right)\tau \right)|z(0)|.
\end{equation}
\par Since $\alpha$ and $\delta$ are fixed, once $c$ is selected, and the exponential decay holds uniformly in the specified initial conditions, it follows that the origin of \eqref{eq:bl_lemma_system} is locally exponentially stable.
We cannot infer global exponential stability, though, because these parameters change once a larger (compact) set of initial conditions is selected.
However, note that $\delta \to \infty$ only if $c \to \infty$, therefore the convergence rate becomes smaller as the initial conditions grow in norm, but it is always non-zero in any compact set.
From these arguments, and the fact that the bounds hold uniformly in time, it follows that the origin of \eqref{eq:bl_lemma_system} is uniformly globally asymptotically stable.
\end{proof}
\subsection{Main Result}
Let $x_{\text{f}} \coloneqq (e, \tilde{\imath}, z)$, then we can state the following stability result, which allows to solve the \textit{mixed controller-observer problem}.
\begin{theorem}
Consider system \eqref{eq:ctrl_obs_error}, parameterized through the positive scalar $\varepsilon$ as shown above.
Denote with $(w(t), x_{\textup{f}}(t), x_{\textup{s}}(t))$ the trajectories of such system, when they exist, for initial conditions $(w(0), x_{\textup{f}}(0), x_{\textup{s}}(0))$.
Let Assumptions \ref{hyp:omega}-\ref{hyp:info} hold.
Then, there exist:
\begin{itemize}
\item an open region $\mathcal{R} \subset \mathbb{S}^1\times\mathbb{R}$, independent of $\chi(\cdot)$, $T^*_{\textup{el}}(\cdot)$, $\dot{T}^*_{\textup{el}}(\cdot)$, and such that $\bar{x}_{\textup{s}} \in \mathcal{R}$;
\item a proper indicator of $\bar{x}_{\textup{s}}$ in $\mathcal{R}$, denoted with $\sigma$;
\item class $\mathcal{KL}$ functions $\beta_{\textup{s}}$, $\beta_{\textup{f}}$;
\end{itemize}
such that, for any positive scalars $\Delta_{\textup{f}}$, $\Delta_{\textup{s}}$, $\delta$, there exist $\varepsilon^* > 0$, $W^*>0$ such that, for all $0 < \varepsilon < \varepsilon^*$ and all initial conditions satisfying $|w(0)| \geq W^*$, $|x_{\textup{f}}(0)| \leq \Delta_{\textup{f}}$, $\sigma(x_{\textup{s}}) \leq \Delta_{\textup{s}}$, the resulting trajectories are forward complete and satisfy, for all $t \geq 0$:
\begin{equation}
\begin{split}
|x_{\textup{f}}(t)| &\leq \beta_{\textup{f}}(|x_{\textup{f}}(0)|, t/\varepsilon) + \delta\\
\sigma(x_{\textup{s}}(t)) & \leq \beta_{\textup{s}}(\sigma(x_{\textup{s}}(0)), t) + \delta.
\end{split}
\end{equation}
\end{theorem}
\begin{proof}
From \cite[Lemma 1]{SCL2019}, it follows that there exist $\mathcal{R}$, $\sigma$ and $\beta_{\text{s}}$ defined above such that the trajectories of system \eqref{eq:red_ord_syst} satisfy, for all $t \geq 0$ and all $x_s \in \mathcal{R}$:
\begin{equation}
\sigma(x_{\text{s}}(t)) \leq \beta_{\text{s}}(\sigma(x_{\text{s}}(0)), t).
\end{equation}
Let $\mathcal{K}\coloneqq\{x_{\text{s}}: \sigma(x_{\text{s}}) \leq \beta_{\text{s}}(\Delta_{\text{s}}, 0) + \delta \}$, then by boundedness of $T^*_{\text{el}}$ it follows that, for all $x_{\text{s}} \in \mathcal{K}$, there exists a positive scalar $I^*$ such that $\| i_{\text{q}}^* \|_{\infty} \leq I^*$.
Apply Lemma \ref{lemma:bl_stability} with the bound $I^*$ to imply the existence of a positive scalar $W^*$ such that, for all $|w(0)| \geq W^*$, the trajectories of system \eqref{eq:ctrl_obs_bl} satisfy, for a class $\mathcal{KL}$ function $\beta_{\text{f}}$:
\begin{equation}
|x_{\text{f}}(\tau)| \leq \beta_{\text{f}}(|x_{\text{f}}(0)|, \tau).
\end{equation}
Due to the regularity properties of system \eqref{eq:ctrl_obs_error}, we can use the result in \cite{TeelNesic} to imply the existence of a positive scalar $\varepsilon^*$ which yields the bounds of the statement.
\end{proof}
Clearly, appropriate choice of $\delta$ ensures an arbitrarily small residual torque tracking error.
We finally summarize the overall controller-observer structure:
\begin{equation}
\begin{split}
&\dot{\hat{\imath}} = L^{-1}\Omega^T(i_{\hat{\chi}})(\hat{\theta} + \beta(i_{\hat{\chi}})) + L^{-1}u_{\hat{\chi}} - \hat{\omega}_\chi\mathcal{J}i_{\hat{\chi}} + k_{\text{p}}(i_{\hat{\chi}} - \hat{\imath})\\
&\dot{\hat{\theta}} = -\frac{\partial \beta}{\partial i_{\hat{\chi}}}(i_{\hat{\chi}}) \left[ L^{-1}\Omega^T(i_{\hat{\chi}})(\hat{\theta} + \beta(i_{\hat{\chi}})) + L^{-1}u_{\hat{\chi}} - \hat{\omega}_\chi\mathcal{J}i_{\hat{\chi}} \right]\\
&\dot{\hat{\zeta}}_\chi = \hat{\omega}_\chi \mathcal{J}\hat{\zeta}_\chi, \quad \dot{\hat{\xi}} = \gamma\hat{h}_1, \quad \hat{\omega}_\chi = |\hat{h}|\hat{\xi} + k_\eta\hat{h}_1\\
&\frac{d}{dt}\begin{pmatrix} w_1\\ w_2 \end{pmatrix} = \begin{pmatrix} 0 & \lambda\\ -\lambda & 0 \end{pmatrix}\begin{pmatrix} w_1\\ w_2 \end{pmatrix} \quad e = \hat{\imath} - \begin{pmatrix}w_1 \\ i_{\text{q}}^* \end{pmatrix}\\
&\beta(i_{\hat{\chi}}) = k_z\begin{pmatrix}
-\frac{|i_{\hat{\chi}}|^2}{2}\\i_{\hat{\chi}}
\end{pmatrix}, \qquad
\begin{split}
\hat{R} &= \begin{pmatrix}1 & 0_{1\times2} \end{pmatrix}\hat{\theta} - (k_z/2)|i_{\hat{\chi}}|^2\\
\hat{h} &= \begin{pmatrix}0_{2\times1} & I_2 \end{pmatrix}\hat{\theta} + k_z i_{\hat{\chi}}
\end{split}
\\
&i_q^* = \frac{2}{3p}\hat{\xi} T^*_{\text{el}} \qquad p_{\text{q}}^* = \frac{2}{3p} \left(\dot{\hat{\xi}} T^*_{\text{el}} + \hat{\xi} \dot{T}^*_{\text{el}} \right)\\
&u_{\hat{\chi}} = -\Omega^T(i_{\hat{\chi}})(\hat{\theta} + \beta(i_{\hat{\chi}})) + L\hat{\omega}_{\chi}\mathcal{J}i_{\hat{\chi}} -Lk_e e + L\begin{pmatrix} \lambda w_2 \\ p_{\text{q}}^*\end{pmatrix}.
\end{split}
\end{equation}

\section{Numerical Results}\label{sec:simulation}

\begin{table}[t!]	
	\begin{center}
		\captionsetup{width=0.7\columnwidth}
		\caption{System parameters}\label{tab:TabParMot}

		\addtolength{\tabcolsep}{-4pt}  
		\begin{tabular}{lr | lr}\hline
			\hline
			{\scriptsize Stator resistance $R$ $[\Omega]$}  &{\scriptsize $0.108$} & {\scriptsize Stator inductance $L$ [$\mu \text{H}$]} & {\scriptsize $30.62$} \\ 
			 {\scriptsize Nominal angular speed [$\text{rpm}$]}  &  {\scriptsize $7000$} & {\scriptsize Rotor magnetic flux $\varphi$ [$\text{mWb}$]}&{\scriptsize $1.309$}\\
			 {\scriptsize Number of pole pairs $p$} &{\scriptsize $12$} & {\scriptsize Load inertia [$\text{Kgm}^2$]} & {\scriptsize $1.4 \times 10^{-4}$} \\ 
			\hline
			\end{tabular}
			\addtolength{\tabcolsep}{4pt}
			  
		\end{center}

\end{table}
For the simulation tests that we present in the following, we adopted as benchmark a motor for electric multirotor UAV propulsion, whose parameters are presented in Table \ref{tab:TabParMot} (corresponding to the commercial PMSM \textit{Tmotor 4006 KV380 }).
The proposed controller-observer was combined with a standard PI speed controller, in a typical cascade structure.
In particular, the proportional and integral gains of the speed controller were set to $k_{\text{p} \omega} = 0.018$, $k_{\text{i}\omega}=0.072$, respectively. 
A first order filter with time constant $\tau=0.1\text{ms}$ was then used to obtain $T^*_{\text{el}}$, $\dot{T}^*_{\text{el}}$.
The gains of the proposed scheme were selected as $k_{\text{p}} = 3.93\times10^3$, $k_e=1.964\times10^3$, while the frequency for the current harmonic injection was set to $\lambda /(2\pi) = 2 \text{kHz}$.
In place of the scalar gain $k_z$, we considered a matrix gain of the form $\diagonal\{0.005, 0.75, 0.75\}$ (the above results still apply in this context with some increased notational burden).
For what concerns the adaptive attitude observer tuning, we considered the linearization of the corresponding dynamics \eqref{eq:red_ord_syst} about the attractor $\bar{x}_{\text{s}}$, with $\chi$ set according to nominal flux and half the nominal speed ($3500 \text{rpm}$).
Then the eigenvalues of the resulting second order linear system (see \cite{SCL2019} for details) have been placed in $(-1\pm (1/3) \text{i})\times 10^2$, by means of gains $k_{\eta}=34.75$, $\gamma = 335.34$ in order to ensure proper time scale separation.
A demanding time-varying speed profile with ramps and sinusoidal waveforms, $\omega^*_{\text{m}}$, has been selected as benchmark to test the system under challenging conditions.
The observer estimates have been initialized to zero, except for $\hat{\xi}$, whose initial value has been set to $802.29 \text{Wb}^{-1}$.
\par Figure \ref{fig:SimTest} shows the results obtained under the aforementioned working conditions.
In plot (a-1) it can be seen that, after an initial transient phase, the estimate of the mechanical speed, $\hat {\omega}/p = |\hat{h}|\hat{\xi}/p$, closely matches the true signal.
Note that we did not include the term $k_\eta \hat{h}_1$, as in the expression of $\hat{\omega}_\chi$, in order to reduce noise sensitivity.
In plot (a-2) the mechanical speed estimation error is obtained as $\tilde{\omega}/p=(\omega - \hat{\omega})/p$, while plot (b-1) shows the the angular estimation error $\tilde{\vartheta} = \atan(\eta_2,\eta_1)$. 
Plot (b-2) shows the parameter $\xi$ and its estimate $\hat{\xi}$, which works correctly and with limited errors during the fast sinusoidal phase of the speed profile.
The fast subsystem variables rapidly converge to the actual signals, as expected from the imposed time scale separation.
The  resistance estimate is portrayed in plot (c-1), along with its actual value.
Note that a fast reconstruction is achieved, even with a zero initial estimate (which is somewhat penalizing as some knowledge about the nominal resistance value is usually available), and kept with minimal mismatches during the most demanding part of the speed profile.
As a result, effective torque estimation and tracking is achieved, as highlighted in plot (c-2).
In turn, the system speed accurately tracks the reference trajectory.
Finally, plots (d-1)-(d-2) depict over the initial $15 {\text{ms}}$ of the simulation the PMSM currents, along with the respective references, and the tracking error $e$, which is quickly steered to zero.
\begin{figure}[h!]
	\centering
	\psfragscanon
	
	\vspace{3pt}
	
	\begin{subfigure}[b]{0.23\textwidth}
	\psfrag{x} [B][B][0.7][0]{\small(a-1) time [s]}
	\psfrag{y} [B][B][0.7][0]{\small$\omega/p$, $\hat{\omega}/p$, $\omega^*_{\text{m}}$ [$\text{rpm}$]}
	\includegraphics[clip = true, width = \textwidth]{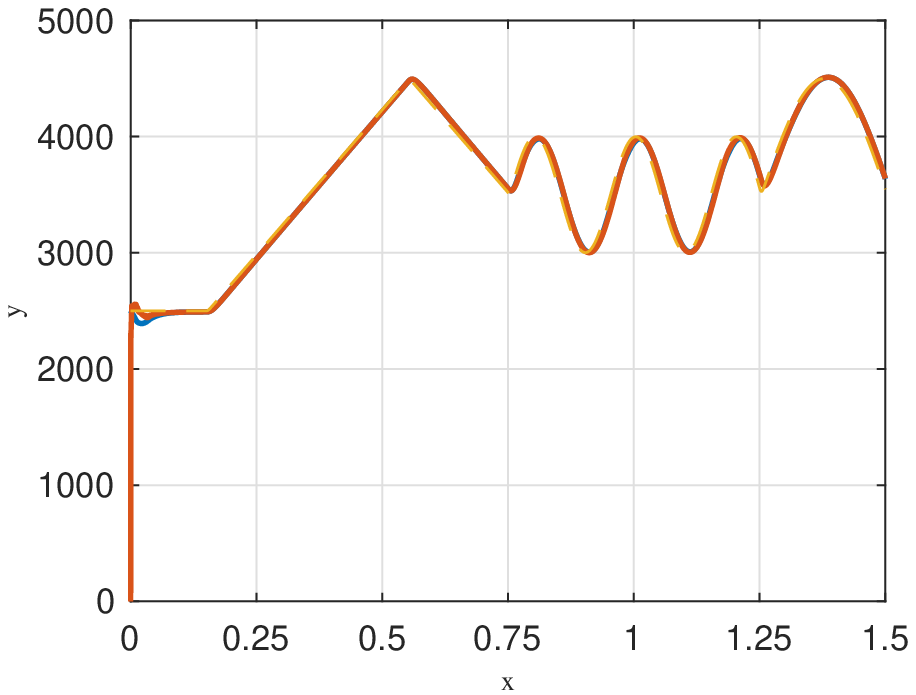}
	
	\vspace{3pt}
	
	\psfrag{x} [B][B][0.7][0]{\small(a-2) time [s]}
	\psfrag{y} [B][B][0.7][0]{\small $\tilde{\omega}/p$ [$\text{rpm}$]}
	\includegraphics[clip = true, width = \textwidth]{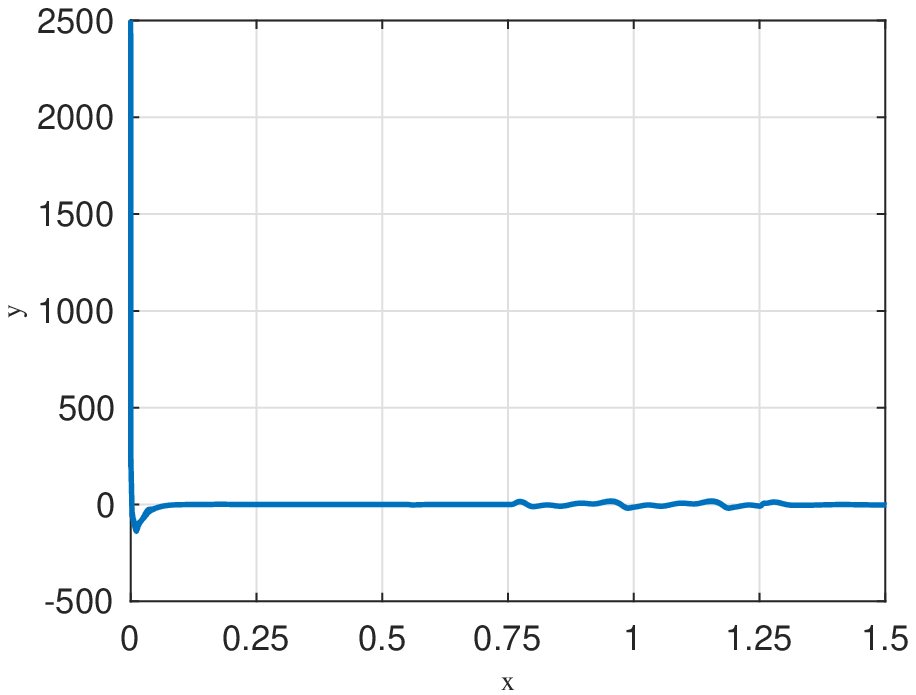}
	\end{subfigure}
	\hspace{0.01\textwidth}
	\begin{subfigure}[b]{0.23\textwidth}
	\psfrag{x} [B][B][0.7][0]{\small(b-1) time [s]}
	\psfrag{y} [B][B][0.7][0]{\small $\tilde{\vartheta}$ [$\text{rad}$]}
	\includegraphics[clip = true, width = \textwidth]{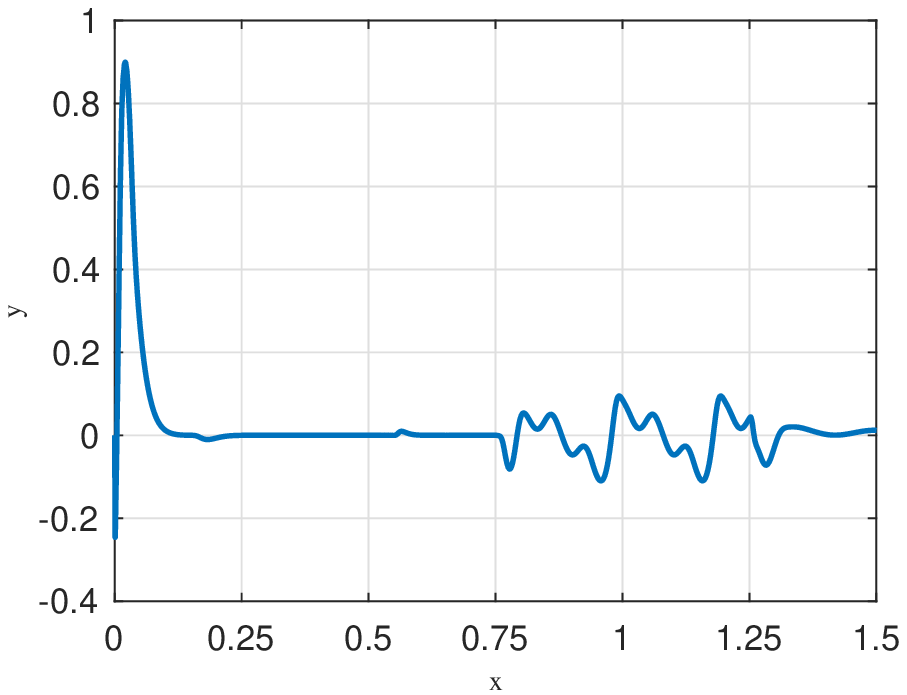}
	
	\vspace{3pt}
	
	\psfrag{x} [B][B][0.7][0]{\small(b-2) time [s]}
	\psfrag{y} [B][B][0.7][0]{\small $\xi$, $\hat{\xi}$ [${\text{Wb}}^{-1}$]}
	\includegraphics[clip = true, width = \textwidth]{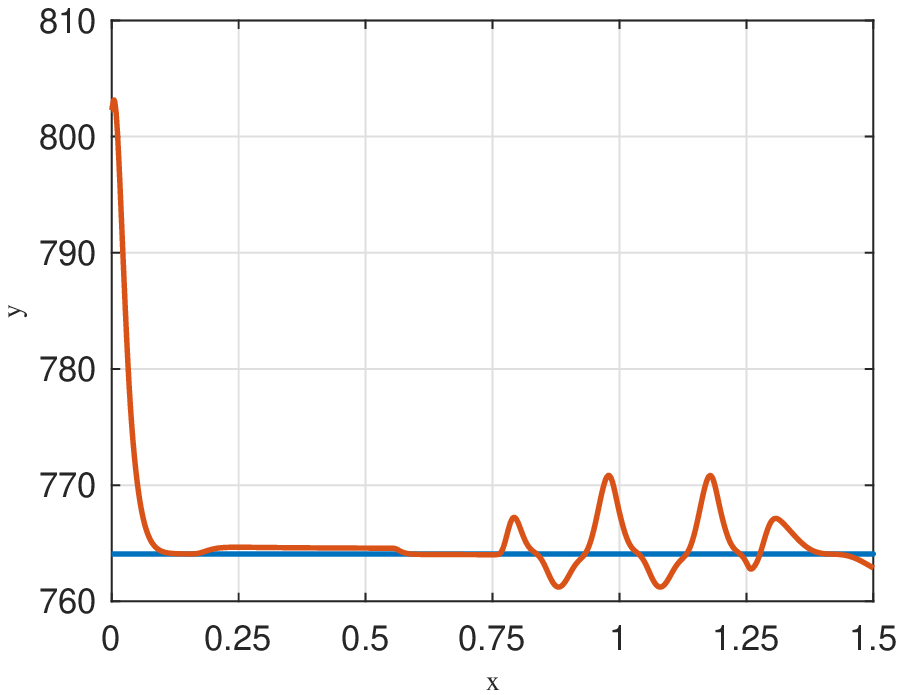}
	\end{subfigure}
	\hspace{0.01\textwidth}
	\begin{subfigure}[b]{0.23\textwidth}
	\psfrag{x} [B][B][0.7][0]{\small(c-1) time [s]}
	\psfrag{y} [B][B][0.7][0]{\small $R, \hat{R}$ [$\Omega$] }
	\includegraphics[clip = true, width = \textwidth]{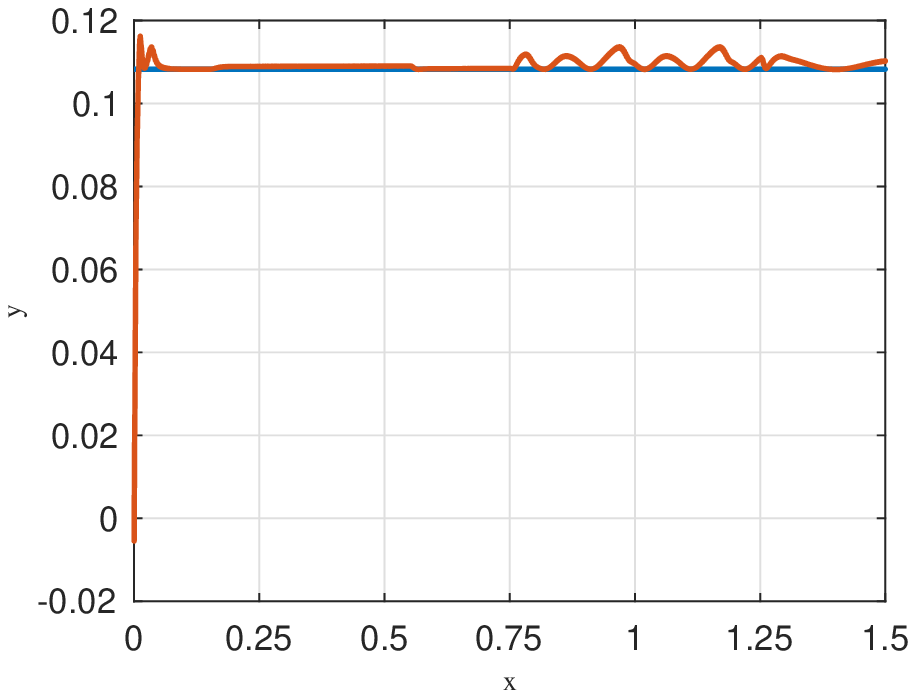}
	
	\vspace{3pt}
	
	\psfrag{x} [B][B][0.7][0]{\small(c-2) time [s]}
	\psfrag{y} [B][B][0.7][0]{\small $T_{\text{el}}, \hat{T}_{\text{el}}, T_{\text{el}}^*$ [${\text{Nm}}$] }
	\includegraphics[clip = true, width = \textwidth]{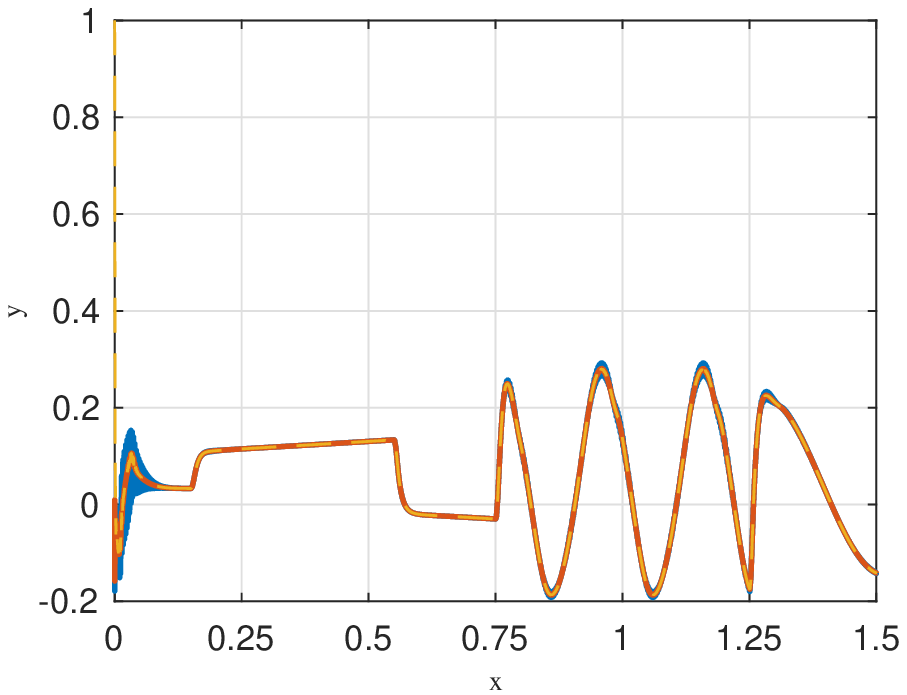}
	\end{subfigure}
	\hspace{0.01\textwidth}
	\begin{subfigure}[b]{0.23\textwidth}
	\psfrag{x} [B][B][0.7][0]{\small(d-1) time [s]}
	\psfrag{y} [B][B][0.7][0]{\small (s) $i_{\hat \chi}$, (d) $i_{\hat{\chi}}^*$ [$A$]}
	\includegraphics[clip = true, width = \textwidth]{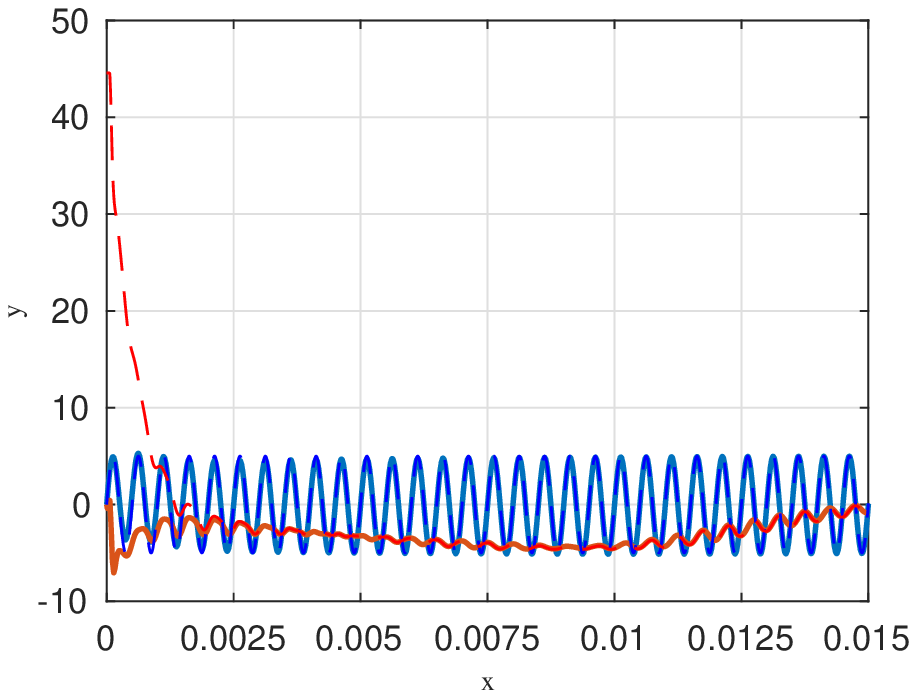}
	
	\vspace{3pt}
	
	\psfrag{x} [B][B][0.7][0]{\small(d-2) time [s]}
	\psfrag{y} [B][B][0.7][0]{\small $e$ [${\text{A}}$] }
	\includegraphics[clip = true, width = \textwidth]{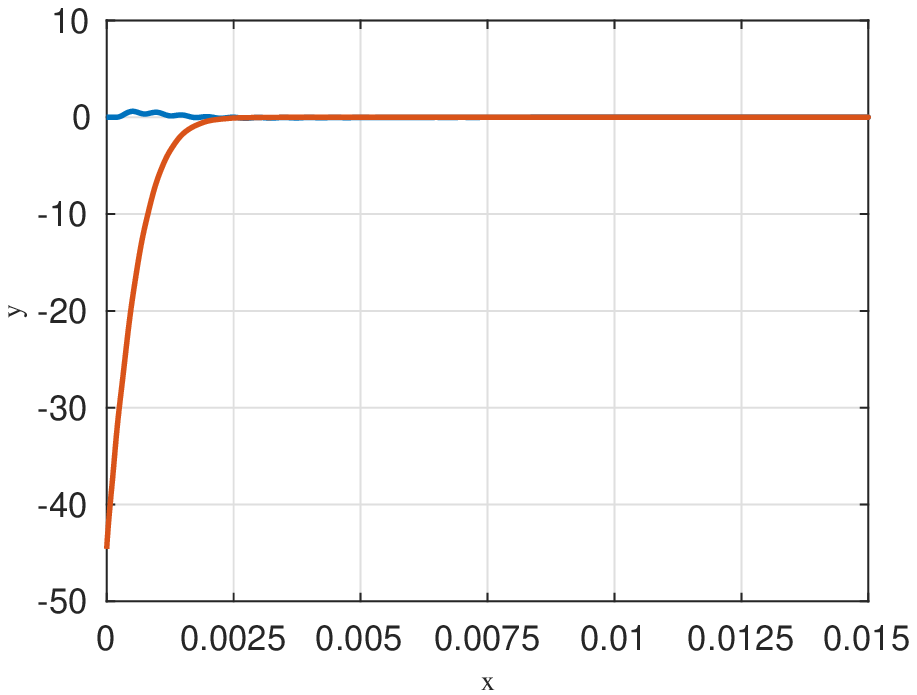}
	\end{subfigure}
	
	\vspace{3pt}
	
	\caption{(a-1): Rotor speed (blue), estimated value (red) and speed reference (dashed yellow). (a-2): mechanical speed estimation error. (b-1): rotor position reconstruction error. (b-2): parameter $\xi$ (blue) and its estimate (red).  (c-1): stator resistance (blue) and its estimate (red). (c-2): torque (blue), estimated signal $\hat{T}_{\text{el}}=(3/2)p( i_{\hat \chi_2}/\hat{\xi})$ (red) and torque reference (dashed yellow). (d-1): current signals $i_{\hat{\chi}}$ (solid) and corresponding current references $i^*_{\hat \chi}$ (dashed), with the first component in blue, the second in red. (d-2): current tracking error $e$, with the first component in blue, the second in red.}
	\label{fig:SimTest}
\end{figure}

\section{Conclusions and Future Works}\label{sec:conclusion}

A mixed sensorless controller-observer for PMSMs was proposed to solve at the same time an observation problem and a torque regulation task.
Opportune signal injection was imposed to ensure observability of the unknown stator resistance, and hence provide formal guarantees of stability and robustness for both the aforementioned objectives.
In addition, the design was proved to work in scenarios of variable speed, with no assumption on the mechanical model.
Numerical simulations were then presented to further validate the structure.
Future efforts will be dedicated to relaxing the assumptions imposed to the rotor speed, and application to other classes of electric motors.

\newpage

\bibliographystyle{ieeetr}
{\bibliography{IFAC_2020_sensorelss_bibliography}}\

\end{document}